\documentclass[conference]{IEEEtran}

\makeatletter

\def\ps@headings{%

\def\@evenhead{\scriptsize\thepage \hfil \leftmark\mbox{}}%

\def\@oddfoot{}%

\def\@evenfoot{}}

\makeatother

\usepackage{algorithmic}

\usepackage{algorithm}
\usepackage{verbatim}
\usepackage{amsmath}
\usepackage{enumerate}
\usepackage{fancybox}

\usepackage{verbatim}
\usepackage{array}
\usepackage{graphicx}
\usepackage{float}
\usepackage[footnotesize]{subfigure}
{\begin{Sbox}\begin{minipage}}%
{\end{minipage}\end{Sbox}\fbox{\TheSbox}}

\newtheorem{theorem}{Theorem}[section]

\begin{document}

\title{A Light-Weight Forwarding Plane \\for  Content-Centric Networks }

\author{J.J. Garcia-Luna-Aceves$^{1,2}$
and Maziar Mirzazad-Barijough$^2$ \\
$^1$Palo Alto Research Center, Palo Alto, CA 94304 \\
$^2$Department of Computer Engineering,
 University of California, Santa Cruz, CA 95064\\
 Email: jj@soe.ucsc.edu, maziar@soe.ucsc.edu }

\maketitle


\begin{abstract}

We present CCN-DART, a more efficient  forwarding approach for content-centric networking (CCN) than named data networking (NDN)  that substitutes Pending Interest Tables (PIT) with Data Answer Routing  Tables (DART) and uses a novel approach to eliminate forwarding loops. The forwarding state required  at each router using  CCN-DART consists of segments of the routes between consumers and content providers that traverse a content router, rather than the Interests that the router forwards towards content providers. Accordingly, the size of a DART is proportional to the number of routes used by Interests traversing a router, rather than the number of Interests traversing a router.  We show that CCN-DART avoids forwarding loops by comparing distances to name prefixes reported by neighbors, even when routing loops exist.  Results of simulation experiments comparing  CCN-DART with NDN  using the ndnSIM  simulation tool show that CCN-DART   incurs 10 to 20 times less  storage overhead.

\end{abstract}


\section{Introduction}

The leading  information-centric networking (ICN) approach can be characterized as {\em Interest-based} and consists of: populating forward information bases (FIB) maintained by routers with routes to name prefixes denoting content, sending content requests (called Interests) for specific content objects (CO) over paths implied by the FIBs, and delivering content along the reverse paths traversed by Interests. 

The original content-centric networking (CCN) proposal 
was the first example of an Interest-based ICN architecture  in which Interests need not be flooded and do not state the identity of the sender. Today, 
named data networking (NDN) \cite{ndn} and CCNx \cite{ccnx} are the prevalent Interest-based ICN approaches. 

Section~\ref{prelim} summarizes the operation of the NDN and CCN forwarding plane. 
Since the introduction of the original CCN proposal,
the research community  (e.g., \cite{ndn,  ndn-fw2, 
ndn-paper}) has assumed  that  per-Interest forwarding state maintained in
Pending Interest Tables (PIT), which is called a ``stateful forwarding plane" in NDN \cite{ndn-fw2},  is needed to allow  Interests and responses to such  Interests (NDOs or negative acknowledgments) to be forwarded without divulging the  sources of  Interests, and that  Interests from different consumers requesting the same content need to be aggregated to attain efficiency.

Section~\ref{sec-rationale} presents the design rationale for {\bf CCN-DART}, a simpler and far more efficient implementation of CCN based on replacing
per-Interest forwarding state with per-route forwarding state.  CCN-DART is an improvement  over \cite{ocean}.

Section~\ref{sec-design}  describes the operation of  CCN-DART. 
Instead of a PIT, a content router  maintains
a data answer routing  table (DART). DARTs allows routers to forward responses to Interests  toward the correct neighbors who requested them, without requiring Interests to state their origins. Routers with local content consumers or routers that support on-path caching maintain a requested content table (RCT) that lists the NDOs stored locally or requested by local consumers.
An Interest in CCN-DART states the name of the requested content and a
destination-and-return token ({\em dart}). The dart is used by forwarding routers to leave a trace of the path traversed by the Interest using local identifiers 
of the previous hop and   the current hop, so that a content object (CO) or a negative acknowledgment (NACK)  can be sent back to the content requestor, without the producer or caching site knowing the source of the Interest. 

To prevent forwarding loops, the FIBs of routers store the hop-count distances to name prefixes from each next hop of a name prefix. A router uses these FIB data to determine whether it can  establish a new DART entry.
Section~\ref{sec-loop} proves that CCN-DART eliminates forwarding loops, which we have shown can exist and remain undetected in NDN and CCNx \cite{ifip2015, ancs2015}. 

Section~\ref{sec-perf} compares the performance  of CCN-DART with NDN 
when
either on-path caching or edge caching is used. CCN-DART and NDN attain the same  end-to-end latencies in retrieving content, and they require similar numbers of forwarded Interests. However, NDN requires an order of magnitude more forwarding entries than CCN-DART. Our results show that maintaining PITs whose sizes grow with the number of Interests handled by  routers is not necessary
to attain correct and efficient forwarding of Interests in a content-centric network.


\section{Elements of NDN and CCNx Operation}
\label{prelim}

Routers in NDN and CCNx use Interests, Data packets, and NACKs (or more generally InterestReturn packets \cite{ccnx}) to exchange content.
An Interest is identified in NDN by the name of the CO requested and a nonce, and in CCNx by the name of the requested CO. A Data packet includes the CO name, a security payload, and the payload itself. A NACK carries the information needed to denote an Interest and a code stating the reason for the response. 

Three data structures are used to process Interests, Data packets, and NACKs:
A Content Store (CS), a FIB, and a PIT.
A CS is a cache for COs indexed by their names. With on-path caching, routers cache the content they receive in response to Interests they  forward.

A FIB is indexed by name prefixes and is populated using content-routing protocols or static routes. The FIB entry  for  a given name prefix  lists the interfaces that can be used to reach the prefix. In NDN, 
a  FIB entry also contains additional information \cite{ndn}.

PITs are used to determine the interfaces to which  Data packets or NACKs  should be sent back in response to Interests, allow Interests to not disclose their sources, and  enable Interest aggregation. 
A PIT entry in NDN   lists the name of a requested  CO,  one or multiple tuples stating a nonce received in an Interest for the CO and the incoming interface where it was received, and a list of the outgoing interfaces over which the Interest was forwarded. 
A PIT entry in CCNx is similar, but no nonces are used.

When a router receives an Interest, it checks whether there is a match in its CS for the CO requested in the Interest. The Interest matching mechanisms differ in NDN and CCNx, with the latter supporting exact Interest matching only.  If a match to the Interest is found, the router sends back a Data packet over the reverse path traversed by the Interest. If no match is found in the CS, the router determines whether the PIT stores an entry for the same content.   

In NDN, 
if the Interest states a nonce that differs from those stored in the  PIT entry for the requested content, then the router ``aggregates" the Interest by adding the incoming interface from which the Interest was received and the nonce to the PIT entry without forwarding the Interest.  If the same nonce in the  Interest is already listed in the PIT entry for the requested CO, the router sends a NACK over the reverse path traversed by the Interest.

Interest aggregation is done in CCNx if an Interest is received from an interface that is not listed in the PIT entry for the requested content.  A repeated Interest received from an interface already listed in the PIT is assumed to be a retransmission and may be forwarded. 
Interests traversing forwarding loops are eventually stopped by means of a hop-limit field included in the Interest and decremented at each hop.

If a router does not find a match in its CS and PIT, the router forwards the Interest along a route  listed  in its FIB for the best prefix match.    
In  NDN, a router can select an  interface  to forward an Interest if it is known that it can bring content  and its performance is ranked higher than other interfaces that can also bring content. The ranking of  interfaces  is done by a router independently of other routers based on information obtained through probing or the control plane \cite{ndn-fw2}.

\vspace{-0.05in}
\section{CCN-DART Rationale  }
\label{sec-rationale}

The design rationale for CCN-DART is based on a few observations reading NDN and the original CCN design. 
First,  the number of  content routers operating in a network is orders of magnitude smaller than the number of COs offered in the network. Hence, maintaining
forwarding state based on the  {\em routes} going through a router--each  used by many Interests--is by nature  orders of magnitude smaller than forwarding state based on the {\em Interests}  traversing a router. A stateful forwarding plane makes sense only if Interest aggregation can reduce the number of forwarded Interests by orders of magnitude, similar to the additional forwarding state incurred with PITs.

Second, we have shown \cite{ifip2015, ancs2015} that Interest loops may go undetected in NDN and CCN when Interest aggregation is supported, which means that Interests need not be answered, even when the requested content is available. Hence, the current NDN and CCN designs must be changed \cite{ancs2015, nof15}.

Third, given that edge caching has been shown to render similar performance results than on-path caching and optimal caching \cite{ali14, faya13},  it is highly likely that  content caching  at the edge makes the occurrence of  Interest aggregation extremely rare. In an Internet environment, the inter-arrival times of Interests for the same content and RTTs between consumers and producers and caches are such that content is available at caches by the time subsequent Interests requesting the same content arrive. 

Fourth, searching a CO in the CS followed by a search of a request for the CO in the PIT is redundant, in that the same CO  name is searched twice. A single table could be used  listing what content is stored  or requested locally.

\section{CCN-DART Operation }
\label{sec-design}

Interests are retransmitted only by the consumers that originated them, rather than routers that relay Interests.
Routers are assumed to know which interfaces are neighbor routers and which are local consumers, and forward  Interests on a best-effort basis. Given that no prior work shows that any Interest matching policy is better than simple exact matching of Interests, we assume that routers use exact Interest matching. 

\subsection{Information Exchanged and Stored}
\label{sec-info}

CCN-DART uses Interests, NACKs and Data packets to support content  exchange among routers.  Our description of these messages addresses only that   information needed to attain correct forwarding, which consists of the names of COs, the hop counts to prefixes,  and {\em destination-and-return tokens} ({\bf darts}). The terms neighbor and interface are used interchangeably.
The name of CO $j$ is denoted by $n(j)$, the name prefix that is the best match for $n(j)$ in a FIB is denoted by $n(j)^*$,  and $S^i_{n(j)^*}$ denotes the set of neighbors of router $i$ considered to be next hops to prefix $n(j)^*$. 
Darts are  local identifiers used to uniquely denote routes established between source and destination routers   over which Interests, Data packets, and NACKs are sent. Accordingly, darts can be very small   (e.g., 32 bits).

An Interest forwarded  by router $k$ requesting CO $j$ is denoted by  $I[n(j),  h^I(k), dart^I(k) ]$. It states the name  of the requested CO ($n(j)$), the hop count ($h^I(k)$) from  $k$ to prefix $n(j)^*$, and the dart  ($dart^I(k)$) used to establish an anonymous route back to the router that originates the Interest.

A Data packet sent in response to  an Interest is denoted by  
$DP[n(j), sp(j),$ $ dart^I(i)  ]$, and  states the name  of the CO requested in the Interest being answered ($n(j)$),  a  security payload ($sp(j)$) used optionally to validate the content object,  the dart ($dart^I(i)$) from the Interest being 
answered,  and the CO  itself. 
A NACK  to  an Interest is denoted by 
$NA[n(j), $ $ \mathsf{CODE}, $ $ dart^I(i)  ]$  and states 
the name  of the CO ($n(j)$),  a code ($\mathsf{CODE}$) indicating the reason why the NACK is sent,  and the dart ($ dart^I(i)$)  from the Interest being 
answered.  Reasons for sending a NACK include: an Interest loop is detected, no route is found towards requested content,  and no content is found.

Router $i$ maintains three tables:
a forwarding information base ($FIB^i$), a data-answer routing table ($DART^i$), and
a  requested-content table ($RCT^i$).
All routers must maintain FIBs and DARTs, and 
only those routers with local users and routers supporting content caching need to maintain an RCT.

A {\it predecessor} of router $i$ for Interests related to name prefix $n(j)^*$ 
is a router that forwards  Interest for COs with names that are best matched by  $n(j)^*$.  
Similarly, a   {\it successor} of router $i$ 
for Interests related to $n(j)^*$ is a router to whom router $i$ forwards Interest regarding COs with names  that are best matched by  name prefix $n(j)^*$.

$FIB^i$ is indexed using name prefixes.  
The entry for prefix $n(j)^*$ consists of a set of  tuples, one for each 
next hop  to prefix $n(j)^*$.
The tuples for prefix $n(j)^*$ are ranked based on their utility for forwarding.
As a minimum, the  tuple for next hop  $q \in S^i_{n(j)^*}$ 
states: 
{\small
\begin{enumerate}
\item
$h(i, n(j)^*, q)$: The distance to  $n(j)^*$ through $q$.
\item
$a(i, n(j)^*, q)$: The nearest anchor of  $n(j)^*$  through $q$.
\end{enumerate}
}

$DART^i$ stores the mappings of predecessors to successors along loop-free paths to name prefixes.
The entry created for Interests  received  from router $p$ (predecessor) and forwarded  to router $s$ (successor) towards a given anchor $a$ of a name prefix  is denoted by
$DART^i(a, p)$ and specifies:
 
{\small
\begin{enumerate}
\item
$a^i ( a, p ) $: The anchor $a$ for which the entry is set. 
\item
$p^i ( a, p ) $: The predecessor $p$ of the path to $a^i ( a, p ) $.
\item
$ pd^i(a, p )$: The {\em predecessor dart}, which equals the dart received in Interests from $p$ forwarded towards $a^i ( a, p )$.
\item
$s^i (a, p ) $: The successor $s$ selected by router  $i$ to forward Interests 
received from $p$ towards  $a^i ( a, p )$. 
\item
$sd^i (a, p) $: The {\em successor dart} included in Interests sent 
by router $i$ towards  $a^i ( a, p )$ through the successor.
 \item
 $h^i (a, p)$: The  hop-count distance to prefix $a$ through successor $s$ when $i$  establishes the DART entry.
\end{enumerate}
}

DART entries can be removed using a least-recently used policy or a maximum lifetime, for example. 
An entry in a DART can remain in storage for long periods of time in the absence of topology changes.  The removal of a DART entry 
simply causes a router to compute a new entry for Interests flowing towards an
anchor of prefixes. 

$RCT^i$ serves as an index of local content as well as local requests for remote content. It   is indexed by the CO names that have been requested by the router.  The entry for CO name $n(j)$  states the name of the CO ($n(j)$),  a pointer to  the local storage where the  CO ($p[n(j)]$) is stored, and a list of zero or more identifiers of local consumers ($lc[n(j)]$) that have requested the CO.  The RCT could be implemented as two separate indexes, one for local content and one for requests for remote content. 

If router $i$ is an anchor of name prefix $n(j)^*$ then it stores all the COs with names that match the name prefix. This is denoted by 
$n(j)^* \in RCT^i$.
If CO  $n(j)$ has been requested by one ore more local consumers and no copy of the CO is yet available, then  
$n(j) \in RCT^i$, $p[n(j)] = nil$,  and $lc[n(j)] \not= \phi$. On the other hand, if router $i$ caches CO  $n(j)$, then $n(j) \in RCT^i$,
$p[n(j)] \not= nil$, and $lc[n(j)] = \phi$.

We have shown \cite{ifip2015, ancs2015} that undetected Interest loops can occur in NDN and CCNx when Interests are aggregated while traversing routing loops resulting from inconsistencies in  FIB entries or inconsistent rankings of  routes at different routers.
CCN-DART uses the same approach we proposed for the 
elimination of undetected Interest loops in the context of  NDN and CCNx to prevent forwarding loops when DART entries are created.

 \vspace{0.05in}
\noindent
{\bf  Dart Entry Addition Rule (DEAR):}  \\
Router $i$  
accepts $I[n(j), h^I(k), dart^I(k)]$ from router $k$ and creates a DART entry for prefix $n(j)^*$ with $k$ as its predecessor  and a router $v \not= k$ as its successor if: 
\vspace{-0.05in}
\[ 
\exists~ v \in  S^i_{n(j)^*} (~ h^I( k)  > h(i, n(j)^*, v) ~)
\]

 \vspace{-0.25in}
\begin{figure}[h]
\begin{centering}
    \mbox{
    \subfigure{\scalebox{.19}{\includegraphics{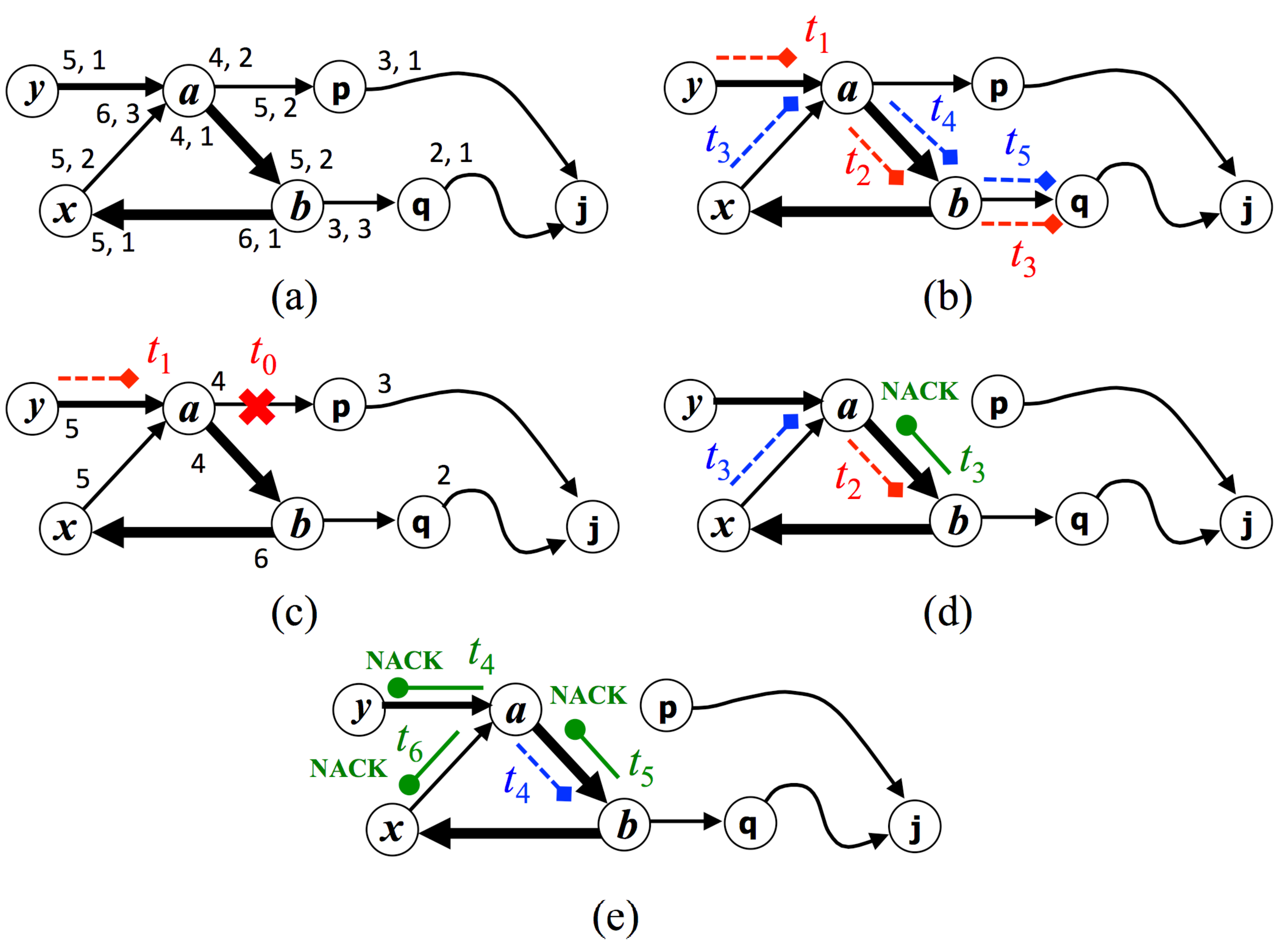}}}
      }
\vspace{-0.1in}
   \caption{CCN-DART avoids forwarding loops
   }
   \label{no-loop}
\end{centering} 
\end{figure}  

 \vspace{-0.1in}
The distance  information that must be stored in the FIBs to implement DEAR can be obtained easily from the control plane. Such content routing protocols as DCR \cite{dcr} and NLSR \cite{nlsr} are able to compute the required minimum-hop distances, which can then be  copied into the FIBs.

Figures~\ref{no-loop}(a) and (b) illustrate how using DEAR prevents 
Interests from traversing loops  when a multi-path routing protocol is used  and 
FIB entries are consistent but  local rankings of multiple routes available at each router  (e.g., NLSR \cite{nlsr}) cause  routing loops.
The pair of numbers next to a node in Figure~\ref{no-loop}(a) indicate the hop count from that router to $n(j)^*$ over an interface and the ranking of the interface according to the FIB of the router. 

Let the triplet 
$(v, h, r)$ denote an interface, its hop count and its ranking.
In Figure~\ref{no-loop}(a), $FIB^a$  states $(b, 4, 1)$, $(p, 4, 2)$,  
$(x, 6, 3)$, and $(y, 6, 4)$;  $FIB^b$ states $(x, 6, 1)$, $(a, 5, 2)$, and
$(q, 3, 3)$; and $FIB^x$ states $(a, 5, 2)$ and 
$(b, 5, 1)$. As Figure ~\ref{no-loop}(b) shows,  router $a$ receives 
$I[n(j), h^I(y) = 5,  dart^I(y)]$ from router $y$ at time $t_1$ 
and  forwards $I[n(j),$ $ h^I(a) = 4,$  $dart^I(a)]$ to $b$ because
$5= h^I(y) > h(a, n(j)^*, b) = 4$ and $b$ is ranked above $p$.  
Similarly, router $b$ receives the Interest at time $t_2$ and accepts it  because $4= h^I(a) > h(b, n(j)^*, q) = 3$.
Router $b$ uses  router $q$ as the next hop for the Interest, because 
$q$ is the highest ranked neighbor satisfying DEAR. 
This example illustrates that, independently of local rankings of multiple routes to prefixes, Interests 
traverse simple paths  by requiring each relaying router to satisfy DEAR.
 
Figures \ref{no-loop}(c) to (e) illustrate how DEAR operates when FIBs are inconsistent due to topology changes. 
Routers $a$ and $b$ update their FIBs at time times $t_0$ and $t_1$, respectively.
We assume that the routing updates have not been processed 
at routers $y$ and  $a$ when they forward Interests at time $t_1$ and $t_2$, respectively.
As shown in Figure \ref{no-loop}(d), router $b$ {\em must} send 
a NACK   to router $a$ 
because it does not have a neighbor with a shorter hop count  to prefix $n(j)^*$ than $h^I(a) = 4$. In turn, router
 $a$ forwards a NACK to router  $y$, and the Interest from $x$ also prompts a NACK from $b$ because DEAR is not satisfied.  Within a finite time after $t_1$, the FIBs of routers are  updated to show that  prefix $n(j)^*$ cannot be reached and Interests from local users for COs in that prefix cannot forwarded by routers $a$, $b$, $x$ and $y$. 
  
By contrast, assuming NDN or CCNx in the same example results in the Interests sent by $y$ and $x$ to be  forwarded along the forwarding loop involving $a$, $b$ and $x$. Router  $a$ aggregates the Interest from $x$, and router $x$ aggregates the Interest from $y$.  Those Interests  must then ``wait  to infinity"   in the PITs until their lifetimes expire or they are otherwise evicted from the PITs. Using nonces in Interests incurs considerable PIT storage overhead. However,  denoting Interests  using only CO names as in CCNx can result in even more Interest-looping problems.
Given the speed with which FIBs are updated to reflect correct distances computed in the control plane, false loop detection using DEAR should be  rare, and it is preferable  by far than storing PIT entries that  expire only  after many seconds without receiving responses. 

\subsection{Maintaining Forwarding State  }

Algorithms~\ref{algo-CCN-DART-create-Interest}  to  4 specify 
the steps  taken by  routers to process and forward Interests,  Data packets, and NACKs.   
The algorithms we present assume that the control plane updates $FIB^i$ to reflect any changes in hop counts  to name prefixes and anchors resulting from the loss of connectivity to one or more neighbors or link-cost changes. 
In addition,  a DART entry is silently deleted when connectivity with the successor or predecessor of the entry  is lost, or  it is not used for a prolonged period of time.

Algorithm \ref{algo-CCN-DART-create-Interest}  shows the steps taken by router $i$ to process Interests received from local consumers. For convenience, content requests from local consumers are assumed to be Interests stating  the name of a CO, together with an empty hop count to content and an empty dart. 

Router $i$ first looks up its RCT to determine if the requested CO is stored locally or a request for the CO  is pending.
If the CO is stored locally, a Data packet is sent back to the user requesting it. 
If a request for the same content is pending, the name of the user is added to the list of customers that have requested the CO. 
Router $i$ sends back a NACK if it is an anchor of name prefix $n(j)^*$ and the specific CO is not found locally, or the CO is remote and no FIB entry exists for a name prefix that can match $n(j)$.

If possible, router $i$ forwards the Interest through the highest ranked neighbor in its FIB for the name prefix matching  $n(j)$.   How such a ranking is done is left unspecified, and can be based on a distributed or local algorithm.

 If a DART entry exists for the selected successor that should receive  the Interest, the existing route is used; otherwise, a new DART entry is created before the Interest is sent. The successor dart assigned to the new DART entry is a locally unique identifier that must be different from all other successor darts being used by router $i$.

\vspace{-0.05in}
{\fontsize{6}{6}\selectfont
\begin{algorithm}[h]
\caption{Processing Interest  from user $c$ at router $i$}
\label{algo-CCN-DART-create-Interest}
{\fontsize{6}{6}\selectfont
\begin{algorithmic}
\STATE{{\bf function}  Interest\_Source}
\STATE {\textbf{INPUT:}  $RCT^i$, $FIB^i$, $DART^i$,  $I[n(j), nil, nil]$;}

\IF{$n(j) \in RCT^i$ }
	\IF{$p[n(j)] \not= nil$ }
		\STATE{
		retrieve CO $n(j)$; 
		send  $DP[n(j), sp(j), nil ]$ to $c$}
	\ELSE
		\STATE{$lc[n(j)]  = lc[n(j)]  \cup c$;  $p[n(j)] = nil$ ~(\% Interest is aggregated)}
	\ENDIF
\ELSE
	\IF{$n(j)^* \in RCT^i$}
		\STATE{send  $NA[n(j),  \mathsf{no~ content}, nil ]$ to $c$ ~(\% $n(j)$ does not exist) }
	\ELSE
		\IF{$n(j)^* \not\in FIB^i$ }
			\STATE{
			send $NA[n(j),  \mathsf{no~ route}, nil ]$ to $c$ ~(\% No route to $n(j)^*$ exists)} 
		\ELSE
			\STATE{create  entry for $n(j)$ in $RCT^i$: ~(\% Interest from $c$ is recorded) \\
		$lc[n(j)] =  lc[n(j)] \cup c$;  $p[n(j)] = nil$; }

			\FOR{{\bf each} $v \in S^i_{n(j)^*}$ {\bf by rank in} $FIB^i$} 
				\STATE{$a = a(i, n(j)^*, v)$; } 
				\IF{ $\exists   DART^i (a, i) ~(~ s^i(a, i) =   v ~)$}
					\STATE{ 
			 		$h^I(i) = h^i(a, i)$;  
					$dart^I(i) = sd^i( a, i)$; \\
					send $I[n(j), h^I(i), dart^I(i) ]$ to  $v$; {\bf return} 		
					}
				\ELSE
					\STATE{create entry $DART^i(a, i)$:  \\
				     compute $SD \not=  sd^i(p, q)  ~\forall ~DART^i( p, q)$; 
					\\
					$pd^i(a, i) =  SD$;  $sd^i(a, i) = SD$; \\
					$p^i(a, i) = i$; $s^i(a, i) = v$; $h^i(a, i) =  h(i, n(j)^*, v)$; \\
					} 
					\STATE{
					 $ h^I(i) = h^i(a, i) $;  
					 $dart^I(i) =   sd^i(a, i) $; \\
					send $I[n(j), h^I(i), dart^I(i) ]$ to  $v$}; 
					 {\bf return} 
				\ENDIF
			\ENDFOR	
		\ENDIF
	\ENDIF
\ENDIF
\end{algorithmic}
}
\end{algorithm}
}

\vspace{-0.05in}
Algorithm~\ref{algo-CCN-DART-Data} outlines the processing of Data packets.  
If the router has local consumers that requested the content, the Data packet is sent to those consumers based on the information stored in $RCT^i$. If the Data packet is received in response to an Interest that was forwarded from router $k$, router $i$ forwards the Data packet 
after swapping the successor dart received in the Data packet for the predecessor dart stored in $DART^i$. Router $i$ stores the data object if  caching is supported.

\vspace{-0.05in}
\begin{algorithm}[h]
\caption{Processing Data packet at router $i$}
\label{algo-CCN-DART-Data}
{\fontsize{6}{6}\selectfont
\begin{algorithmic}
\STATE{{\bf function} Data Packet\_Handling}
\STATE{\textbf{INPUT:} $DART^i$,  $RCT^i$,  
$DP[n(j), sp(j),  dart^I(q) ]$; }
\STATE{{\bf [o]} verify $ sp(j)$;}
\STATE{{\bf [o]} {\bf if} verification fails {\bf then} discard $DP[n(j), sp(j), dart^I(q) ]$}

\IF{$\exists DART^i(a, k) ~(~ dart^I(q)= sd^i(a, k)  \wedge p^i(a, k) = i ~)$ \\~~~(\% router $i$ was the origin of the Interest) }
	\FOR{{\bf each} $c \in lc[n(j)]$}
		\STATE{send $DP[n(j), sp(j), nil ]$ to  $c$; $lc[n(j)] = lc[n(j)] - \{ c \}$}
	\ENDFOR
\ENDIF
\IF{$\exists DART^i(a, k) ~(~ dart^I(q)= sd^i(a, k) \wedge p^i(a, k) = k \in N^ i  ~)$}
	\STATE{
	(\% Data packet can be forwarded to $k$:)\\
	$dart^I(i) = pd^i (a, k)$;  
	send $DP[n(j), sp(j),  dart^I(i) ]$ to  $k$}
\ENDIF

\STATE{
		{\bf [o]} {\bf if} no entry for  $n(j)$ exists in $RCT^i$ {\bf then} \\
		~~~~~~~~create  $RCT^i$ entry for $n(j)$: $lc[n(j)] = \emptyset$  \\~~~~~{\bf end if}\\
	{\bf [o]} store CO in local storage;   $p[n(j)] =$ address of CO in local storage
	} 

\end{algorithmic}}       
\end{algorithm}

\vspace{-0.05in}
Algorithm~\ref{algo-CCN-DART-nack} states the steps taken to handle NACKs, which are similar to the forwarding steps taken after receiving a Data packet. 
Router $i$ forwards the NACK to local consumers  if it was the origin of the Interest, or  to a neighbor router $k$  if it has a DART entry  with a successor dart  matching the dart stated in the NACK.

\vspace{-0.05in}
\begin{algorithm}[h]
\caption{Process NACK at router $i$}
\label{algo-CCN-DART-nack}
{\fontsize{6}{6}\selectfont
\begin{algorithmic}
\STATE{{\bf function} NACK\_Handling}
\STATE {\textbf{INPUT:} $DART^i$,  $RCT^i$,  
$NA[n(j), \mathsf{CODE},  dart^I(q)]$; 
}
\IF{$\exists DART^i(a, k) ~(~ dart^I(q)= sd^i(a, k)  \wedge p^i(a, k) = i ~)$ \\~~~(\% router $i$ was the origin of the Interest) }
	\FOR{{\bf each} $c \in lc[n(j)]$}
		\STATE{send $NA[n(j), \mathsf{CODE}, nil ]$ to  $c$; $lc[n(j)] = lc[n(j)] - \{ c \}$
		}
	\ENDFOR
	\STATE{
	delete entry for $n(j)$ in $RCT^i$
	}
\ENDIF

\IF{$\exists DART^i(a, k) ~(~ dart^I(q)= sd^i(a, k) \wedge p^i(a, k) = k \in N^ i ~)$}
	\STATE{
	(\% NACK can be forwarded to router $k$:)\\
		$dart^I(i) = pd^i (a, k)$; 
		send $NA[n(j), \mathsf{CODE}, dart^I(i) ]$ to  $k$}
\ENDIF
\end{algorithmic}}
\end{algorithm}

\vspace{-0.2in}
\begin{algorithm}[h]
\caption{Processing Interest  from router $k$ at router $i$}
\label{algo-CCN-DART-Interest}
 {\fontsize{6}{6}\selectfont
\begin{algorithmic}
\STATE{{\bf function} Dart\_Swapping}
\STATE {\textbf{INPUT:}  $RCT^i$, $FIB^i$, $DART^i$,  
$I[n(j), h^I(k),  dart^I(k)]$;}

\IF{$n(j) \in RCT^i \wedge p[n(j)] \not= nil$ }
		\STATE{
		retrieve CO $n(j)$; 
		send  $DP[n(j), sp(j), dart^I(k) ]$ to $k$}
\ELSE
	\STATE{(\% $n(j) \not\in RCT^i \vee p[n(j)] = nil$ )}
	\IF{$n(j)^* \in RCT^i$}
		\STATE{send  $NA[n(j),  \mathsf{no~ content}, dart^I(k)]$  to $k$ }
	\ELSE 
		\IF{$n(j)^* \not\in FIB^i$ }
			\STATE{
			send $NA[n(j),  \mathsf{no~ route}, dart^I(k)]$ to $k$ }
		\ELSE
			\IF{$\exists  DART^i (a, k)  ~(~ pd^i(a, k) =   dart^I(k) $ )}
				\STATE{ 
				$h^I(i) = h^i(a, k )$;  
				$dart^I(i) = sd^i( a, k)$;  \\
				send $I[n(j), h^I(i),  dart^I(i) ]$ to  $s^i(a, k ) $ 		
				}
			\ELSE
				\FOR{{\bf each} $v \in S^i_{n(j)^*}$ {\bf by rank in} $FIB^i$} 
					\IF {$ h^I(k) > h(i, n(j)^*, v) $   ~(\% DEAR is satisfied) }
					\STATE{$a = a(i, n(j)^*, v) $;}
					\STATE{create entry $DART^i( a, k)$: \\
					compute $SD \not=  sd^i(p, q)  ~\forall ~DART^i( p, q)$; \\
					$h^i(a, k) =  h(i, n(j)^*, v)$; \\
					$p^i(a, k) = k$; 
					$pd^i(a, k) =  dart^I(k)$; \\
					$s^i(a, k) = v$; 
					$sd^i(a, k) = SD$;} 
					\STATE{create Interest: \\
  					$ h^I(i) = h^i(a, k) $; 
					$dart^I(i) =   sd^i(a, k) $; \\
					send $I[n(j), h^I(i),   dart^I(i) ]$ to  $v$; \\
					{\bf return} }
					\ENDIF
				\ENDFOR ~(\% Interest may be traversing a loop)
				\STATE{send $NA[n(j),  \mathsf{loop},  dart^I(k) ]$ to $k$} 
			\ENDIF
		\ENDIF
	\ENDIF
\ENDIF

\end{algorithmic}
}
\end{algorithm}

\vspace{-0.05in}
Algorithm~\ref{algo-CCN-DART-Interest}  shows the steps needed to process an Interest received from a neighbor router $k$. 
If the requested content is cached locally, a Data packet  is sent back. 
As in Algorithm \ref{algo-CCN-DART-create-Interest}, router $i$ sends back a NACK if it is an anchor of $n(j)^*$ and the CO with name $n(j)$ is not found locally, or the CO is remote and no FIB entry exists for $n(j)^*$.
In contrast to Algorithm \ref{algo-CCN-DART-create-Interest},  Interests received from other routers are not aggregated.

If the Interest must be forwarded and an entry exists in $DART^i$  with router $k$ as the predecessor and $dart^I(k)$ as the predecessor dart, then DEAR has been satisfied by a previous Interest from $k$ over the existing path that $k$ is requesting to use.   Accordingly, router $i$ simply swaps $dart^I(k)$  for the successor dart stated in the entry in $DART^i$ 
and forwards the Interest. 
Alternatively, if  no DART entry exists with $k$ as a predecessor and $dart^I(k)$ as the predecessor dart, router $i$ tries to  find a neighbor that satisfies DEAR. 
The highest-ranked router $v$
satisfying DEAR is selected as the successor for the Interest and router $i$ creates a 
successor dart different from any other successor darts in $DART^i$, stores an entry with  $v$ and the new successor dart, and forwards the Interest to $v$.
If DEAR is not satisfied, then router $i$ sends a NACK back to router $k$.

\subsection{CCN-DART Forwarding Example }
\label{sec-example}

Figure \ref{sdarts} illustrates how darts are used to label Interests and associate Interests with Data packets and NACKs.
In the example, routers $a$ , $b$, and $x$ have local consumers originating  Interests, and those Interests are  assumed to request
COs with names that are best matched with name prefixes  for which   router $d$ is an anchor.

The arrowheads in the links of the figure denote the next hops stored in the FIBs of routers, and  $y(i)$ denotes the $i$th dart in $DART^y$. 
The figure shows the DART entries maintained  at all  routers for 
name prefixes for which router $d$ is an anchor, and the   RCT entries stored at routers $a$, $b$, and $x$. 
Consumers {\em A, C, N,} and {\em P} request the same CO with name $n(j)$, and router $a$ aggregates their requests and needs to send only one Interest for $n(j)$ towards $d$. Similarly, it aggregates the Interests from consumers {\em A, C,} and {\em Q}. Similar Interest aggregation of local requests occur at routers $b$ and $x$. 

\vspace{-0.20in}
\begin{figure}[h]
\begin{centering}
    \mbox{
    \subfigure{\scalebox{.23}{\includegraphics{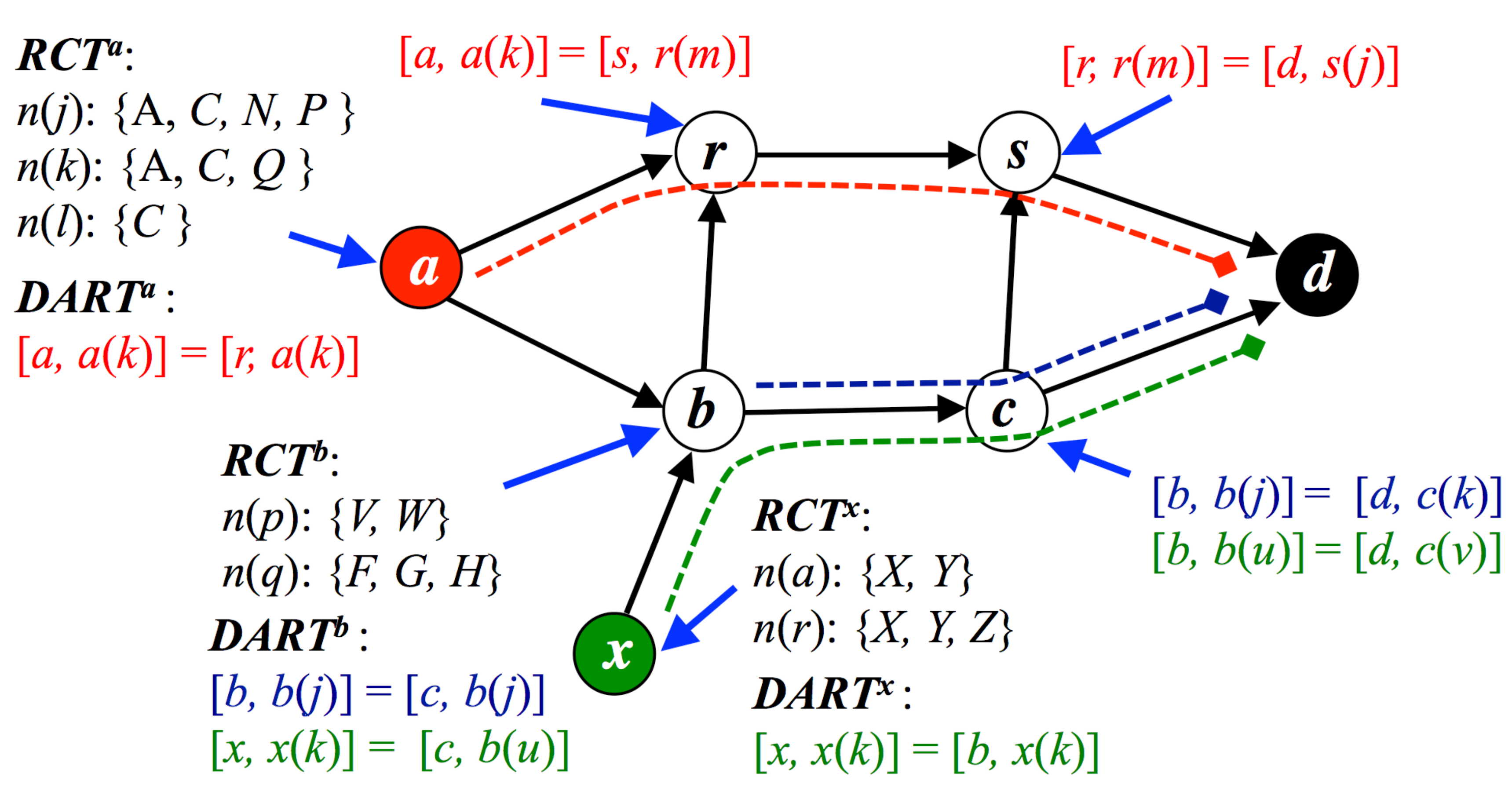}}}
      }
\vspace{-0.20in}
   \caption{Interest forwarding in CCN-DART.
   }
   \label{sdarts}
\end{centering} 
\end{figure}  

\vspace{-0.1in}
Router $a$ 
forwards  Interests intended for anchor $d$ to neighbor $r$, and routers $b$ and $x$ forwards their Interests to neighbors $s$ and $c$, respectively.  
Routers $a$, $r$, and $s$  establish the following mappings in their DARTs:
$ [ a; a(k) ] \leftrightarrow [r;  a(k)]$ at $a$,
$ [ a; a(k) ] \leftrightarrow [s;  r(m)]$ at $r$, and
$ [ r; r(m) ] \leftrightarrow [d;  s(j)]$ at $s$.
These mappings  denote the route ($a$, $r$, $s$, $d$)
uniquely. Similarly,  
routers establish the DART mappings shown in the figure that denote the routes ($x$, $b$, $c$, $d$) and  ($b$, $c$, $d$). 

All the Interests from consumers local to routers $a$, $b$, and $x$ regarding COs with names  in prefixes for which $d$ is an anchor can be routed towards $d$ using the same few darts shown. Given that 
a Data packet  or NACK specifies
the successor dart  stated the Interest it answers, 
Data packets and NACKs can be forwarded correctly from $d$ (or a router  caching the requested CO along the way to $d$) to routers $a$, $b$, or  $x$ unambiguously. 
In turn, routers $a$, $b$, and $x$ can determine how to deliver the responses to local consumers based on the the RCT entries mapping each CO requested with the names of the customers that requested them.

\section{Correctness of  CCN-DART }
\label{sec-loop} 

The following  theorem shows that CCN-DART prevents  Interests from being propagated along loops,  independently of whether the 
topology is static or dynamic or the FIBs are consistent or not.   
To discuss the correctness of Interest forwarding in CCN-DART, we say that
a forwarding  loop of $h$ hops for a CO with name $n(j)$ occurs  when  Interests requesting the CO are forwarded  by routers along a cycle $L = $   $\{ v_1 , v_2 , ..., v_h , v_1 \}$,
such that router $v_k$ receives an Interest for CO $n(j)$ from $v_{k-1}$ 
and forwards the Interest  to $v_{k+1}$, with $1 \leq k \leq h$, $v_{h+1} = v_1$, and $v_{0} = v_h$. 

\begin{theorem}
\label{theo3}
Interests  cannot traverse forwarding loops in a network in which  CCN-DART is used.
\end{theorem}

\begin{proof}
Consider a network in which CCN-DART is used.  Assume for the sake of contradiction that  routers in a forwarding  loop  $L$ of $h$ hops  $\{ v_1 , $ $v_2 , ..., $ $v_h , v_1 \}$  send 
Interests for $n(j)$ along $L$, with no router in $L$  detecting the incorrect forwarding of any of the Interests sent over the loop.

Given that  $L$ exists by assumption, $v_k \in L$ must send 
$I[n(j), h^I(v_{k}),  dart^I(v_k)]$ to router $v_{k+1} \in L$ for $1 \leq k \leq h - 1$, and $v_h \in L$ must send $I[n(j),$ $ h^I(v_{h}),$  $ dart^I(v_{h}) ]$ to router $v_{1} \in L$. 

For $1 \leq k \leq h - 1$, let $h(v_k, n(j)^*)^L$ denote the value of $h^I (v_k)$ when router  $v_k$   sends $I[n(j),$ $ h^I(v_{k}),$  $ dart^I(v_k) ]$  to router $v_{k+1}$,  
with $h(v_k, $ $n(j)^*)^L $ $=$  
$h(v_k, n(j)^*, v_{k+1})$. 
Let $h(v_h, n(j)^*)^L$ denote 
the value of $h^I (v_h)$ when router $v_h$ sends
$I[n(j),$ $ h^I(v_{h}),$ $dart^I(v_h) ]$ to router $v_1 \in L$, with $h(v_h, $ $n(j)^*)^L$ $ =$  $ h(v_h, n(j)^*, $ $v_1)$. 

Because no router in $L$ detects the incorrect forwarding  of an Interest 
and forwarded Interests are not  aggregated in CCN-DART, each router in $L$ must send its own Interest as a result of the Interest it receives from the previous hop in $L$. This implies that router
$v_k \in L$ must accept $I[n(j), h^I(v_{k-1}),$ $dart^I(v_{k-1}) ]$ 
for $1 \leq k < h $, and  router $v_1 \in L$ must accept $I[n(j), h^I(v_{h}),$  $dart^I(v_{h}) ]$.

If router $v_k$ sends Interest $I[n(j), $ $h^I(v_{k}),$  $dart^I(v_{k}) ]$ to router $v_{k+1}$ as a result of receiving 
$I[n(j), h^I(v_{k-1}),$  $dart^I(v_{k - 1}) ]$ from router $v_{k-1}$, then 
it must be true that  $h^I(v_{k - 1}) $ $ > $ $h(v_{k}, n(j)^*)^L $ $=$ $ h^I(v_{k}) $ for $1 < k \leq h$.
Similarly, if router
$v_1$ sends $I[n(j), h^I(v_{1}),$  $ dart^I(v_{1}) ]$ to router $v_{2}$ as a result of receiving 
$I[n(j),$  $h^I(v_{h}), $ $dart^I(v_{h}) ]$ from  router $v_{h}$, then  $h^I(v_{h}) > $ $h(v_{1},$ $ n(j)^*)^L = $ $h^I(v_{1}) $.

It follows from the above argument that, for $L$ to exist and be undetected when each router in the loop uses DEAR to create  DART entries,  
it must be true that $h^I(v_{h}) > h^I(v_{1}) $
and $h^I(v_{k - 1}) >  h^I(v_{k}) $ for $1 < k \leq h$.
However, this is a contradiction, because it implies that $h^I(v_{k}) >  h^I(v_{k})$ for $1 \leq k \leq h $. Therefore, the theorem is true. \end{proof}

\section{Performance Comparison}
\label{sec-perf}

We compare the forwarding entries needed to forward Interests and responses in NDN and CCN-DART  using simulation experiments  based on implementations of NDN and CCN-DART in the ndnSIM simulation tool \cite{ndnsim}.  The NDN implementation was used without modifications, and CCN-DART was implemented in the ndnSIM tool following Algorithms 1 to 4. 
The network topology consists of 200 nodes distributed uniformly in a 100m $\times$ 100m area and nodes with distance of 12m or less are connected with point-to-point links of delay 15ms. The data rates of the links are set to 1Gbps to eliminate the effects that a sub-optimal implementation of  CCN-DART or NDN may have on the results. All nodes are producers and consumers at the same time, which is the worst-case scenario for CCN-DART; and consumers generate object requests with a Zipf distribution with  parameter $\alpha = 0.7$.
Producers are assumed to publish 1,000,000 different COs. Simulation runs consist of request rates of 50 to 2000 interests per router second, and represents the aggregate rate of Interests from  local consumers at each router.

The assumption that each 
router is locally attached to a content producer and   local users requesting content constitutes the worst-case scenario for CCN-DART compared to NDN, because it results in many more DART entries. In a realistic deployment, only a small subset of the total number of routers in the network are attached to local producers of content.

We considered  on-path caching and edge caching. For the case of on-path caching, every router on the path traversed by a Data packet from the producer to the consumer caches the CO in its  content store. On the other hand, with edge caching, only the router directly connected to the requesting  consumer  caches the resulting CO.

\subsection{Size of PITs and DARTs}

Figure \ref{sizePitVSDart} shows the average size and standard deviation of the number of entries in PITs used in NDN and  the number of entries in DARTs and RCTs used in CCN-DART as a function of the total content-request rates per router. The vales shown for RCTs represent only the number of local pending Interests; the number of COs cached locally  is not shown, given that the number of such entries would be very large and would be the same for NDN and CCN-DART.

\vspace{-0.15in}
\begin{figure}[h]
\begin{centering}
    \mbox{
    \subfigure{\scalebox{.36}{\includegraphics{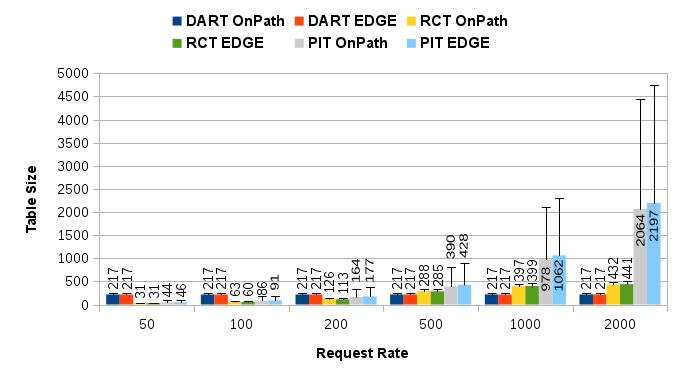}}}
    }
\vspace{-0.35in}
   \caption{Size of PITs, DARTs and RCTs  }
   \label{sizePitVSDart}
\end{centering} 
\end{figure}  

\vspace{-0.05in}
As the figure shows, the size of PITs grows dramatically as the rate of content requests increases, which is expected given that PITs maintain per-Interest forwarding state. By contrast, the size of DARTs remains  constant with respect to the content-request rates. 
The small average size of RCTs compared to the average size of PITs  indicates that  
the average size of a PIT is dominated by the number of Interests a router forwards from other routers. 

For low request rates, the average number of entries in a DART is actually larger than in a PIT. This is a direct consequence of the fact that a PIT entry is deleted immediately after an Interest is satisfied, while a
DART entry is kept  for long periods of time (seconds) in our implementation,  independently of whether or not it is  used to forward Interests.  Given  the small sizes of  DARTs, the cost of maintaining DART entries that may not be used at light loads is more than compensated by the significant reduction in forwarding state derived from many Interests being forwarded using existing DART entries  at higher request rates. This should be the case in real deployments, where the  number of routers that are origins of routes to prefixes is  much smaller than the total number of routers. However, optimizing the length of time that a DART entry lasts as a function of its perceived utility for content forwarding is an area that deserves further study. 

As the total content-request rate per router increases, the size of a PIT can be  more than 10 to 20 times the size of a DART, because a given DART entry is used for many Interests in CCN-DART, while NDN requires a different  PIT entry for each Interest. 
It is also interesting to see the effect of on-path caching compared to edge-caching. 
The average size of DARTs is independent of  where content is being cached, and  on-path caching in NDN does not make a significant difference in the size of a PIT compared to edge-caching. 

\vspace{-0.10in}
\begin{figure}[h]
\begin{centering}
    \mbox{
    \subfigure{\scalebox{.35}{\includegraphics{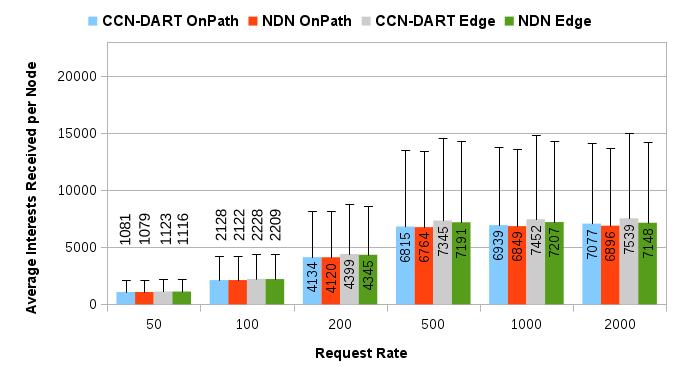}}}
    }
\vspace{-0.28in}
   \caption{Number of Interests received by routers}
   \label{Intcnt}
\end{centering} 
\end{figure}  

\vspace{-0.28in}
\begin{figure}[h]
\begin{centering}
    \mbox{
    \subfigure{\scalebox{.34}{\includegraphics{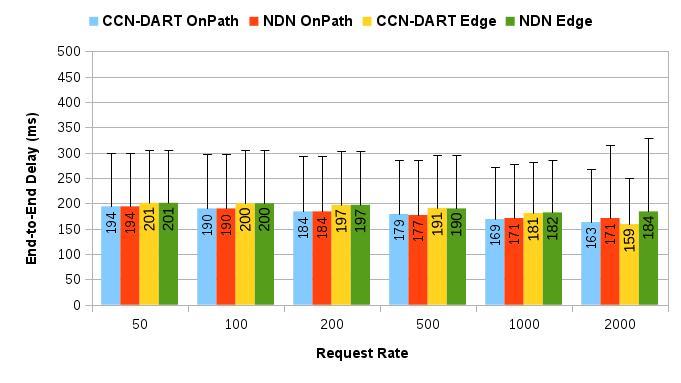}}}
    }
\vspace{-0.28in}
   \caption{Average end-to-end delays}
   \label{delayCCN-DARTvsNDN}
\end{centering} 
\end{figure}

\subsection{Interest Traffic and   End-to-End Delays}

Figure \ref{Intcnt} shows the average number of  interests received by each router in NDN and CCN-DART as a function of the content request rates for on-path caching and edge caching.  The number of Interests received in CCN-DART is larger than the corresponding number for NDN.
However, it is clear from the figure that the average numbers of Interests received by each router in NDN and CCN-DART are almost the same for all request rates. 

The benefit of on-path caching is apparent for both NDN and CCN-DART, but appears 
slightly more pronounced for the case of CCN-DART. This should be expected, because CCN-DART does not aggregate Interests and on-path caching results in fewer Interests being forwarded.

Figure \ref{delayCCN-DARTvsNDN} shows the average end-to-end delay for NDN and CCN-DART as a function of   content-request rates for on-path caching and edge caching.  As the figure shows, the average delays for NDN and CCN-DART are  essentially  the same for all content-request rates. This should be expected, given that in the experiments the routes in the FIBs are static and loop-free,  and the number of Interests processed by routers is similar.

\section{Conclusions}

We  introduced CCN-DART, 
an efficient approach to  content-centric networking that supports Interest forwarding without revealing the sources of Interest and with no need to maintain forwarding state on a per-Interest basis.
CCN-DART replaces PITs with 
DARTS, which establish forwarding state for each route traversing the router over which many Interests are multiplexed, rather than establishing state for each different Interest using  routes traversing the router.

A novel approach  was introduced to eliminate forwarding loops even when routing-table loops exist due to inconsistent states in the FIBs. The additional information needed to implement the new forwarding rule at each router is the distance reported by each next hop to each name prefix. 
The results of simulation experiments based on ndnSIM  show that CCN-DART is far more efficient than NDN. It incurs less than an order of magnitude storage overhead at high loads.


 {\fontsize{6}{6}\selectfont

  }

\end{document}